\documentclass[preprint,12pt]{elsarticle}

\usepackage{amssymb}
\usepackage{amsmath}

\usepackage{longtable,blkarray,graphicx,here,ulem,color,mathtools,comment}
\usepackage{mathtools,float}

\newdefinition{defn}{Definition}
\newtheorem{thm}{Theorem}
\newtheorem{lem}{Lemma}
\newtheorem{prop}{Proposition}
\newdefinition{rmk}{Remark}
\newproof{pr}{Proof}

\allowdisplaybreaks

\journal{Discrete Mathematics}

\begin{document}

\begin{frontmatter}



\title{A group structure arising from Grover walks on complete graphs with self-loops and its application}


\author[label1]{Tatsuya Tsurii\corref{cor1}}
\ead{t3tsuri23@rsch.tuis.ac.jp}

\author[label2]{Naoharu Ito}
\ead{naoharu@cc.nara-edu.ac.jp}

\cortext[cor1]{Corresponding author}

\affiliation[label1]{organization={Department of Informatics, Faculty of Informatics, Tokyo University of Information Sciences},
            addressline={Wakabaku Onaridai 4-1}, 
            city={Chiba},
            postcode={265-8501}, 
            state={Chiba},
            country={Japan}}

\affiliation[label2]{organization={Department of Mathematical Education, Nara University of Education},
            addressline={Takabatake-cho}, 
            city={Nara},
            postcode={630-8528}, 
            state={Nara},
            country={Japan}}

\begin{abstract}
This paper introduces a group-theoretic framework to analyze the algebraic structure of the Grover walk on a complete graph with self-loops. 
We construct a group generated by the Grover matrix and a diagonal matrix whose entries are powers of a complex root of unity. 
We then characterize the resulting quotient group, which is defined using a subgroup formed by commutators involving these matrices. 
We show that this quotient group is isomorphic to a finite cyclic group whose structure depends on the parity of the number of vertices.
This group-theoretic characterization reveals underlying symmetries in the time evolution of the Grover walk and provides an algebraic framework for understanding its periodic behavior.
\end{abstract}



\begin{keyword}
Quantum walk \sep Grover walk \sep Group structure \sep Complete graph \sep Periodicity 

\MSC[2020] 05C25 \sep 15A30 \sep 81Q99 \sep 39A23
\end{keyword}

\end{frontmatter}




\section{\label{sec:intr}Introduction}
Quantum walks, the quantum analogues of classical random walks, have emerged as  powerful tools in quantum computation, mathematical physics, and graph theory (see \cite{AAKV,AKST,BGKLW,HKSS,K,MW,SJK} and references therein). Aharonov et al.~\cite{AAKV} introduced quantum walks as a computational framework, and Kempe~\cite{K} provided a comprehensive survey of their algorithmic and structural properties.

Among various types of quantum walks, the Grover walk is particularly notable for its role in quantum search algorithms and its rich spectral behavior. 
The periodicity of the Grover walk has been investigated on several graph classes, including distance-regular graphs~\cite{Y}, bipartite regular graphs~\cite{Ku}, and complete graphs with self-loops~\cite{IMT}, primarily through spectral and combinatorial techniques.

This paper introduces a different perspective by employing a group-theoretic framework to analyze the algebraic structure underlying the Grover walk on a complete graph with self-loops and $n$ vertices. 
Our approach constructs a group $\mathcal{K}$ generated by two
matrices: the Grover matrix $G$, which governs the time evolution of the walk, and a diagonal matrix $S=\text{diag}[1, \omega, \omega^2, \dots, \omega^{n-1}]$, where $\omega=e^{\frac{2\pi i}{n}}$.
We then analyze the structure of this group via a subgroup $\mathcal{H} \subset \mathcal{K}$ formed by commutators involving $G$ and powers of $S$, which gives rise to the quotient group $\mathcal{K}/\mathcal{H}$.

Our main contribution is the precise characterization of $\mathcal{K}/\mathcal{H}$. 
We prove that $\mathcal{K}/\mathcal{H}$ is a finite cyclic group whose structure depends on the parity of $n$. 
This result provides a purely algebraic explanation for the periodicity of the Grover walk. 
Specifically, we use the structure of $\mathcal{K}/\mathcal{H}$ to determine the smallest positive integer $m$ such that $(S^j G)^m$ equals the identity element of $\mathcal{H}$ for all $j = 0, 1, \dots, n-1$. 
We show that this integer $m$ is $2n$, which corresponds to the period of the Grover walk.

The paper is organized as follows. In Section~\ref{sec:group} we introduce the Grover walk on complete graphs with self-loops and define groups arising from
the Grover walk. Section~\ref{sec:main} presents our main results, including the characterization of $\mathcal{K}/\mathcal{H}$ and the derivation of the periodicity of the Grover walk.

\section{A group arising from the Grover walk} \label{sec:group}

We consider the Grover walk on a complete graph with a self-loop at each vertex $v_j$ ($j = 1, 2, \dots, n$), where $n$ is a positive integer greater than 1. 
Figure~\ref{fig:lena} shows the complete graph with four vertices. 
\begin{figure}[H]
\centering
\includegraphics[scale=1.0, bb=0 0 99 97]{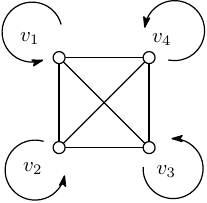} 
\caption{\label{fig:lena}The complete graph with four vertices}
\end{figure}
\noindent 
Note that the edges of each self-loop are directed, but the choice of direction is arbitrary. 
Each vertex $v_j$ has an $n$-dimensional complex vector $\psi_{t}(v_j)$ 
$(t = 0, 1, 2, \dots)$, called the probability amplitude vector. 

Let $P_0, P_1, \dots, P_{n-1}$ be $n \times n$ matrices given by 
\begin{align*}
P_0 &= 
\begin{bmatrix}
\frac{2}{n}-1 & \frac{2}{n} & \frac{2}{n} & \cdots & \frac{2}{n} \\
0 & 0 & 0 & \cdots & 0 \\
0 & 0 & 0 & \cdots & 0 \\
\vdots & \vdots & \vdots & \ddots & \vdots \\
0 & 0 & 0 & \cdots & 0
\end{bmatrix}, \ 
P_1 = 
\begin{bmatrix}
0 & 0 & 0 & \cdots & 0 \\
\frac{2}{n} & \frac{2}{n}-1 & \frac{2}{n} & \cdots & \frac{2}{n} \\
0 & 0 & 0 & \cdots & 0 \\
\vdots & \vdots & \vdots & \ddots & \vdots \\
0 & 0 & 0 & \cdots & 0
\end{bmatrix}, \dots, \\
P_{n-1} &= 
\begin{bmatrix}
0 & 0 & \cdots & 0 & 0\\
0 & 0 & \cdots & 0 & 0 \\
\vdots & \vdots & \ddots & \vdots & \vdots \\
0 & 0 & \cdots & 0 & 0 \\
\frac{2}{n} & \frac{2}{n} & \cdots & \frac{2}{n} & \frac{2}{n}-1 \\
\end{bmatrix}. 
\end{align*}
The matrix $G = P_0 + P_1 + \cdots + P_{n-1}$ is called the Grover matrix of degree $n$. 
Note that $G$ is involutory; that is, $G^2 = I_n$, where $I_n$ denotes the identity matrix of size $n$. 

The Grover walk on a complete graph is a type of quantum walk whose time evolution is given by 
\begin{equation} \label{eq:time_evol}   
\varPsi_{t+1} = U  \varPsi_{t} \quad \text{for } \ t = 0, 1, 2, \dots 
\end{equation}
where $U$ is an $n^2 \times n^2$ matrix defined by 
\[ 
U =  
\begin{bmatrix}
P_0 & P_1 & P_2 & \cdots & P_{n-1}\\
P_{n-1} & P_0 & P_1 & \cdots & P_{n-2}\\
P_{n-2} & P_{n-1} & P_0 & \cdots & P_{n-3}\\
\vdots & \vdots & \vdots & \ddots & \vdots \\
P_1    & P_2    &  P_3  & \cdots & P_0
\end{bmatrix}, 
\]
and $\varPsi_t = [ \psi_{t}(v_1)  \ \psi_{t}(v_2) \ \cdots \ \psi_{t}(v_n) ]^T$ 
(see also \cite{IMT}). 

If all solutions of \eqref{eq:time_evol} are periodic and $p$ is 
the least common multiple of their periods, then the Grover walk is said to be  periodic with period $p$. 

Let $A$ be the $n \times n$ discrete Fourier matrix defined by
\begin{align*}
A = \frac{1}{\sqrt{n}} \begin{bmatrix}
1      & 1            & 1               & \cdots & 1 \\
1      & \omega       & \omega^2        & \cdots & \omega^{n-1}\\
1      & \omega^2     & \omega^4        & \cdots & \omega^{2(n-1)}\\
\vdots & \vdots       & \vdots          & \ddots & \vdots \\
1      & \omega^{n-1} & \omega^{2(n-1)} & \cdots & \omega^{(n-1)^2}
\end{bmatrix},
\end{align*}
where $\omega=e^{\frac{2\pi i}{n}}$. 
In our previous paper \cite{IMT}, we showed that the evolution matrix $U$ is diagonalizable by the matrix $A \otimes I_n$ as follows:  
\begin{align} 
(A \otimes I_n)^{-1} U (A \otimes I_n) = \text{diag} [G, SG, \dots, S^{n-1}G], 
\label{eq:diagU}
\end{align}
where 
\begin{align*}
S = \begin{bmatrix}
1 &      0 & 0 & \cdots & 0 \\
0 & \omega & 0 & \cdots & 0 \\
0 & 0      & \omega^2 & \cdots & 0 \\
\vdots & \vdots & \vdots & \ddots & \vdots \\
0 & 0 & 0 & \cdots & \omega^{n-1}
\end{bmatrix}, \ 
G = 
\begin{bmatrix}
\frac{2}{n}-1 & \frac{2}{n} & \frac{2}{n} & \cdots & \frac{2}{n} \\[2pt]
\frac{2}{n} & \frac{2}{n}-1 & \frac{2}{n} & \cdots & \frac{2}{n} \\[2pt]
\frac{2}{n} & \frac{2}{n} & \frac{2}{n}-1 & \cdots & \frac{2}{n} \\[2pt]
\vdots & \vdots & \vdots & \ddots & \vdots \\[2pt]
\frac{2}{n} & \frac{2}{n} & \frac{2}{n} & \cdots & \frac{2}{n}-1 
\end{bmatrix}. 
\end{align*}
Note that $S^n = I_n$ 
and 
\begin{align}
 &GS^j \ne S^j G \ \text{for} \ j = 1,2, \dots, n-1. \label{eq:GS_inv}
\end{align}

We define the group $\mathcal{K}$ generated by the matrices $S$ and $G$, and 
the subgroup $\mathcal{H} \subset \mathcal{K}$ generated by 
the commutators 
\begin{align*}
[S^j,G] \coloneqq S^j G S^{-j} G^{-1} = S^j G S^{-j} G 
\quad \text{for } j=1,2,\dots,n-1.
\end{align*} 
It is shown that $\mathcal{H}$ is a normal subgroup of $\mathcal{K}$ (see 
Proposition~\ref{prop:normal} in Section~\ref{sec:main}). 

The purpose of this paper is to determine the structure of 
the quotient group $\mathcal{K}/\mathcal{H}$ and 
to show the periodicity of the Grover walk by using the structure.

\section{Main results} \label{sec:main}
Let $\mathbb{Z}_n$ be the cyclic group of order $n$.
Our main results are stated in the following two theorems. 

\begin{thm} \label{thm:structure}
The following assertions hold for $\mathcal{K}/\mathcal{H}$. 
\begin{enumerate}
\item $\mathcal{K}/\mathcal{H} \cong \mathbb{Z}_n \times \mathbb{Z}_2$ 
if $n$ is even. 
\item $\mathcal{K}/\mathcal{H} \cong \mathbb{Z}_n$ if $n$ is odd. 
\end{enumerate} 
\end{thm}

\begin{thm} \label{thm:period}
$2n$ is the smallest positive integer $m$ such that $G^m$, $(SG)^m$, $(S^2G)^m, \dots, (S^{n-1}G)^m$ are all equal to the identity element in $\mathcal{H}$.  
\end{thm}

Recalling \eqref{eq:diagU}, Theorem~\ref{thm:period} implies that $2n$ is the smallest positive integer $m$ such that $U^m = I_{n^2}$.
Therefore, we obtain that the Grover walk on the complete graph 
with $n$ vertices is periodic with period $2n$.
This result yields an alternative proof of Theorem 2 in \cite{IMT}.

We present lemmas and propositions necessary for the proofs of 
our main results.  

\begin{lem} \label{lem:comm}
Let $j \in \{1,2,\dots, n-1 \}$. Then the matrix $A^{-1}[S^j, G]A$ is diagonal, with its $(1,1)$-th and $(j+1,j+1)$-th elements equal to $-1$, while all other elements are equal to $1$. 
\end{lem}

\begin{pr}
By direct computation, 
\begin{align*}
&\text{\hspace{10mm} $(n-j+1)$-th column \hspace{27mm}$(j+1)$-th column} \\
&\text{\hspace{27.7mm}$\downarrow$} \text{\hspace{66mm}$\downarrow$}\\[-3pt]
A^{-1}S^j A   &= 
\begin{bmatrix}
 &        &   & 1 &        & \\
 &  \text{\huge{0}}       &   &   & \ddots & \\
 &        &   &   &        &1 \\
1&        &   &   &        & \\
 & \ddots &   &   &  \text{\huge{0}}       & \\
 &        & 1 &   &        & 
\end{bmatrix},\ 
A^{-1}S^{-j} A   = 
\begin{bmatrix}
 &        &   & 1 &        & \\
 &  \text{\huge{0}}       &   &   & \ddots & \\
 &        &   &   &        &1 \\
1&        &   &   &        & \\
 & \ddots &   &   &  \text{\huge{0}}       & \\
 &        & 1 &   &        & 
\end{bmatrix},\\[3mm]
A^{-1}GA &=\begin{bmatrix}
1          &   0         &   \cdots    &0 \\
 0         &  -1          &   \cdots  & 0 \\
 \vdots &    \vdots & \ddots &  \vdots \\
  0       &  0 &   \cdots    &-1  \\
\end{bmatrix}.
\end{align*}
Hence, 
\begin{align*}
A^{-1}[S^j,G]A &=  A^{-1}S^j A \cdot A^{-1}GA \cdot A^{-1} S^{-j}A 
\cdot A^{-1}G A \\[5pt]
&\text{\hspace{32mm}$(j+1)$-th column} \\
&\text{\hspace{41mm}$\downarrow$}\\[-5pt]
&=\begin{bmatrix}
-1&   &       & & & & & \\
  &  1 &       & & & \text{\huge{0}}&  &\\
  &   & \ddots & & & & &\\
  &   &       &1 & & & &\\
  &   &       & &-1 & & &\\
  & \text{\huge{0}}  &       & & &1 & &\\
  &   &       & & & &\ddots& \\
  &   &       & & & & &1
\end{bmatrix}.
\end{align*}
\hspace{\fill}$\Box$
\end{pr}

Let $\mathbb{N}$ be the set of natural numbers. 
By Lemma \ref{lem:comm}, it is straightforward to verify that the following lemma holds.

\begin{lem} \label{lem:inv}
The commutator $[S^j,G]$ is involutory; that is, $[S^j,G]^2=I_n$ for $j \in \mathbb{N} \cup \{0\}$ and 
\begin{align*} 
 &[S^j,G^k]=[G^k, S^j] \ \text{for } j,k \in \mathbb{N} \cup \{0\}. 
\end{align*}
\end{lem}

The following propositions are derived from Lemmas \ref{lem:comm} and \ref{lem:inv}. 

\begin{prop} 
The following assertion holds for $\mathcal{H}$. 
\begin{align} \label{eq:strH}
\mathcal{H} \cong \underbrace{\mathbb{Z}_2 \times \mathbb{Z}_2 \times \cdots \times 
\mathbb{Z}_2}_{n-1~times}
\end{align}  
\end{prop}

\begin{pr}
It follows from $S^n = I_n$ that 
\[
[S^j,G]=[S^k,G] \text{ for } j\equiv k\bmod n, \ j,k=1,2,\dots.
\]
By Lemma \ref{lem:comm}, we have    
\begin{align*}
[S^j,G] \ne [S^k,G] \text{\ for $j \ne k,\ j,k=1,2,\dots, n-1$,} 
\end{align*}
and 
\begin{align*}
[S^j,G][S^k,G] = [S^k,G][S^j,G] \text{\ for $j,k=1,2,\dots, n-1$.} 
\end{align*} 
Since $[S^j, G]$ is involutory for $j = 1,2, \dots, n-1$, 
we have
\begin{align} \label{eq:strH-2}
\mathcal{H} = \{ [S,G]^{m_1}[S^2,G]^{m_2} \cdots [S^{n-1}, G]^{m_{n-1}} \ 
| \ m_1, m_2, \dots, m_{n-1} =0, 1 \}.
\end{align}
Hence, the order of $\mathcal{H}$ is $2^{n-1}$, and we obtain \eqref{eq:strH}. 
\hspace{\fill}$\Box$
\end{pr}

\begin{prop} \label{prop:normal}
$\mathcal{H}$ is a normal subgroup of $\mathcal{K}$. 
\end{prop}

\begin{pr}
By Lemma~\ref{lem:inv}, we obtain for $j = 1, 2, \dots, n-1$,  
\begin{align*}
S [S^j,G]S^{-1} 
&= S [G, S^j] S^{-1}\\
&= S (G S^{-1}G)(GSG) [G, S^j] S^{-1}\\
&=[S,G]GSG (G S^j G S^{-j}) S^{-1} \\
&=[S,G]GS^{j+1}G S^{-(j+1)}\\
&=[S,G][G, S^{j+1}] \\
&=[S,G][S^{j+1},G] \in  \mathcal{H},
\end{align*}
and 
\begin{align*}
G [S^j,G] G 
&= G S^j G S^{-j} G G \\
&= [G,  S^j] \\
&=[S^j,G] \in \mathcal{H}.
\end{align*}
Therefore, $\mathcal{H}$ is a normal subgroup of $\mathcal{K}$. 
\hspace{\fill}$\Box$
\end{pr}

Since $\mathcal{H}$ is a normal subgroup of $\mathcal{K}$, 
we can consider the quotient group $\mathcal{K}/\mathcal{H}$. 
The following Lemmas \ref{lem:sets} and \ref{lem:H} are used 
to clarify the structure of $\mathcal{K}/\mathcal{H}$. 

\begin{lem} \label{lem:sets}
The following assertion holds. 

\begin{align} \label{eq:sets}
S^j G^k \mathcal{H} = G^k S^j \mathcal{H} \ \text{for} \ 
j, k \in \mathbb{N} \cup \{0\} .
\end{align}
\end{lem}

\begin{pr}
Let $j,k \in \mathbb{N} \cup \{0\}$.
Then it follows from Lemma~\ref{lem:inv} and $G^2=I_n$ that  
\begin{align*} 
S^j G^k [S^{-j}, G^k] 
&= S^j G^k S^{-j}G^k S^j G^{-k} \\ 
&= [S^j, G^k] S^j G^{-k} \\
&= [G^k, S^j] S^j G^{-k} \\
&= G^k S^j G^{-k} S^{-j}S^j G^{-k} \\ 
&= G^k S^j (G^2)^{-k} \\
&= G^k S^j . 
\end{align*}
Noting that $[S^{-j}, G^k] \in \mathcal{H}$, we have $S^j G^k [S^{-j}, G^k] \in S^j G^k \mathcal{H} \cap G^k S^j \mathcal{H} \ne \emptyset$, which implies that \eqref{eq:sets} holds. 
\hspace{\fill}$\Box$
\end{pr}

\begin{lem} \label{lem:H}
The following assertions hold:  
\begin{enumerate}
\item $S^j \notin \mathcal{H}$ for $j=1,2,\dots,n-1$. 
\item $S^j G \notin \mathcal{H}$ for $j=1,2,\dots,n-1$. 
\item $G \in \mathcal{H}$ if $n$ is odd, and 
$G \notin \mathcal{H}$ if $n$ is even. 
\end{enumerate}
\end{lem}

\begin{pr}
Let $j \in \{ 1, 2, \dots, n-1 \}$. 

(i) Suppose that $S^j \in \mathcal{H}$. 
Then there exist $m_1, m_2, \dots, m_r \in \{1,2,\dots,n-1\}$ such that 
$$
S^j = [S^{m_1},G][S^{m_2},G]\cdots[S^{m_r},G]. 
$$
Noticing that $G[S^k,G]G = [S^k,G]$ for $k = 1,2,\dots,n-1$ as shown 
in the proof of Proposition \ref{prop:normal}, we have  
\begin{align*}
GS^j &=G[S^{m_1},G][S^{m_2},G]\cdots[S^{m_r},G]\\
&=G[S^{m_1},G]GG[S^{m_2},G]GG \cdots GG[S^{m_r},G]GG \\
&= [S^{m_1},G][S^{m_2},G]\cdots[S^{m_r},G] G \\
&= S^j G.
\end{align*}
This contradicts \eqref{eq:GS_inv}. 

(ii) 
Suppose that $S^j G \in \mathcal{H}$. 
Then there exist $m_1, m_2 \dots, m_r \in \{1,2,\dots,n-1\}$ such that 
$$
S^j G= [S^{m_1},G][S^{m_2},G]\cdots[S^{m_r},G]. 
$$
Thus, we have  
\begin{align*}
G(S^j G)&=G[S^{m_1},G][S^{m_2},G]\cdots[S^{m_r},G]\\
&=G[S^{m_1},G]GG[S^{m_2},G]GG \cdots GG[S^{m_r},G]GG \\
&= [S^{m_1},G][S^{m_2},G]\cdots[S^{m_r},G] G \\
&= (S^j G) G \\
&= S^j, 
\end{align*}
which implies that $G S^j = S^j G$. 
This contradicts \eqref{eq:GS_inv}.  

(iii) 
Assume that $n$ is odd. Then it follows from Lemma~\ref{lem:comm} that   
\begin{align*}
&A^{-1}[S,G][S^2,G] \cdots [S^{n-1},G]A \\
&= A^{-1}[S,G]AA^{-1}[S^2,G] \cdots AA^{-1}[S^{n-1},G]A \\
&= \begin{bmatrix}
1          &   0         &   \cdots    &0 \\
 0         &  -1          &   \cdots  & 0 \\
 \vdots &    \vdots & \ddots &  \vdots \\
  0       &  0 &   \cdots    &-1  \\
\end{bmatrix} \\
&=A^{-1}GA.  
\end{align*}
Thus, we have $G=[S,G][S^2,G] \cdots [S^{n-1},G]$, which implies that 
$G \in \mathcal{H}$. 

Next, assume that $n$ is even. Then define 
\[
W = A^{-1} [S^{m_1},G][S^{m_2},G]\cdots[S^{m_r},G] A,
\]
where $m_1, m_2, \dots, m_r \in \{1,2,\dots,n-1\}$, and 
$w_{jj}$ denotes the $(j, j)$-th element of $W$ for $j=1,2,\dots, n$. 
If $w_{jj} = -1$ for $j=2, 3, \dots, n$, $r$ is odd. 
Then, $w_{11} = -1$. 
Hence, by Lemma \ref{lem:comm} there do not exist $m_1, m_2, \dots, m_r \in \{1,2,\dots,n-1\}$ satisfying $G=[S^{m_1},G][S^{m_2},G]\cdots[S^{m_r},G]$, 
which implies $G \notin \mathcal{H}$. 
\hspace{\fill}$\Box$
\end{pr}

Using Lemmas~\ref{lem:sets} and \ref{lem:H}, we can prove the following proposition. 

\begin{prop} \label{prop:comp}
The following assertions hold: 
\begin{enumerate}
\item If $n$ is even, $\{ I_n, S, S^2, \dots, S^{n-1}, G, SG, S^2G, 
\ \dots, 
S^{n-1}G \}$ is a complete system of representatives for $\mathcal{K}/\mathcal{H}$. 
\item If $n$ is odd,  $\{ I_n, S, S^2, \dots, S^{n-1} \}$ is a complete system of representatives for $\mathcal{K}/\mathcal{H}$. 
\end{enumerate}
\end{prop}

\begin{pr}
Let $g \in \mathcal{K}$ be expressed as $g=g_1^{m_1} g_2^{m_2} \cdots g_r^{m_r}$,  
where $g_j \in \{S, G\}$ and $m_j \in \mathbb{Z}$ for $j=1,2,\dots, r$. 
Since $S^n=I_n$ and $G^2=I_n$, Lemma~\ref{lem:sets} shows that
there exist $k \in \{ 0,1,2,\dots, n-1 \}$ and $l \in \{0, 1\}$ 
such that  
\begin{align*}
g\mathcal{H} 
&= g_1^{m_1} g_2^{m_2} \cdots g_r^{m_r} \mathcal{H} \\
&= S^k G^l \mathcal{H} . 
\end{align*}
Hence, we have 
\begin{align} \label{eq:KH} 
\mathcal{K}/\mathcal{H} = 
\bigcup_{j=0}^{n-1} S^j \mathcal{H} 
\cup \bigcup_{j=0}^{n-1} S^j G \mathcal{H} .
\end{align}

To prove that $\mathcal{H}, S\mathcal{H}, S^2\mathcal{H},\dots, 
S^{n-1}\mathcal{H}$ are mutually disjoint, 
suppose that there exist $j, k \in \{ 0, 1, 2, \dots, n-1 \}$ with $j \ne k$  such that 
\[
S^j \mathcal{H} \cap S^k \mathcal{H} \not = \emptyset.
\] 
Then there exist $h_1, h_2 \in \mathcal{H}$ such that 
$S^j h_1 = S^k h_2$, 
and hence $S^{j-k} =  h_2 h_1^{-1} \in \mathcal{H}$.   
This contradicts (i) in Lemma~\ref{lem:H}. 
Thus, $\mathcal{H}, S\mathcal{H}, S^2\mathcal{H},\dots,
S^{n-1}\mathcal{H}$ are mutually disjoint.

Assume that $n$ is odd. Then \eqref{eq:KH} and (iii) in 
Lemma~\ref{lem:H} follow 
\[
\mathcal{K}/\mathcal{H} = 
\bigcup_{j=0}^{n-1} S^j \mathcal{H} .
\] 
Therefore, we obtain (ii). 

Assume that $n$ is even. 
Suppose that there exist $j, k \in \{ 0, 1, 2, \dots, n-1 \}$ with $j \ne k \ (j<k)$ such that $S^j G \mathcal{H} = S^k G \mathcal{H}$. 
Then Lemma~\ref{lem:sets} shows that 
\[
\mathcal{H} = GS^{k-j} G \mathcal{H} = S^{k-j} G^2 \mathcal{H}= S^{k-j} \mathcal{H}, 
\]
which implies $S^{k-j} \in \mathcal{H}$. This contradicts 
(i) in Lemma~\ref{lem:H}. 
Next, suppose that there exist $j, k \in \{ 0, 1, 2, \dots, n-1 \}$ such that 
$S^j \mathcal{H} = S^k G \mathcal{H}$. 
Then Lemma~\ref{lem:sets} shows that
\[
\mathcal{H} = S^{k-j} G \mathcal{H},  
\]
which implies $S^{k-j}G \in \mathcal{H}$. 
If $j \ne k$,  this contradicts (ii) in Lemma~\ref{lem:H}. 
If $j=k$, we have $G \in \mathcal{H}$, which contradicts 
(iii) in Lemma~\ref{lem:H}. 
Therefore, we obtain (i). This completes the proof.
\hspace{\fill}$\Box$
\end{pr}

We are now ready to prove our main results. 

\bigskip
\noindent
\textbf{Proof of Theorem~\ref{thm:structure}.}
Let $g_1$ and $g_2$ be generators of $\mathbb{Z}_n$ and $\mathbb{Z}_2$, respectively. 

Assume that $n$ is even. 
It follows from Proposition \ref{prop:comp} that the quotient group $\mathcal{K}/\mathcal{H}$ has order $2n$ and 
\[
\mathcal{K}/\mathcal{H} = \{ \mathcal{H}, S\mathcal{H}, S^2\mathcal{H},\dots, 
S^{n-1}\mathcal{H}, G\mathcal{H}, SG\mathcal{H}, S^2G\mathcal{H}, \ \dots, S^{n-1}G\mathcal{H} \}. 
\]
Define a map $\varphi$ from $\mathcal{K}/\mathcal{H}$ to $\mathbb{Z}_n \times \mathbb{Z}_2$ by 
\[
\varphi(S^k G^l \mathcal{H}) = (g_1^k, g_2^l) \text{ for } k=0,1,2,\dots, n-1,\ l=0, 1. 
\]
By Lemma \ref{lem:sets}, we see that for $k_1, k_2 \in \{0,1,2,\dots,n-1\}$ and $l_1, l_2 \in \{0, 1\}$, 
\begin{align*}
\varphi((S^{k_1} G^{l_1} \mathcal{H})(S^{k_2} G^{l_2} \mathcal{H}))
&= \varphi(S^{k_1} G^{l_1} S^{k_2} G^{l_2} \mathcal{H}) \\
&= \varphi(S^{k_1+k_2} G^{l_1+l_2} \mathcal{H}) \\
&= (g_1^{k_1+k_2}, g_2^{l_1+l_2}) \\
&=(g_1^{k_1}, g_2^{l_1})(g_1^{k_2}, g_2^{l_2}) \\
&=\varphi(S^{k_1} G^{l_1} \mathcal{H})\varphi(S^{k_2} G^{l_2} \mathcal{H}), 
\end{align*}
which implies that $\varphi$ is a group homomorphism. 
Since $\varphi$ is bijective, we obtain (i). 

Next, assume that $n$ is odd. 
It follows from Proposition \ref{prop:comp} that the quotient group $\mathcal{K}/\mathcal{H}$ has order $n$ and
\[
\mathcal{K}/\mathcal{H} = \{ \mathcal{H}, S\mathcal{H}, S^2\mathcal{H},\dots, 
S^{n-1}\mathcal{H} \} . 
\]
Define a map $\psi$ from $\mathcal{K}/\mathcal{H}$ to $\mathbb{Z}_n$ by 
\[
\psi(S^k \mathcal{H}) = g_1^k \text{ for } k=0,1,2,\dots, n-1. 
\]
Since $\psi$ is a group homomorphism and bijective, we obtain (ii). 
This completes the proof. 
\hspace{\fill}$\Box$

\bigskip
\noindent
\textbf{Proof of Theorem~\ref{thm:period}.} 
Since we have $G^n \in \mathcal{H}$ by (iii) in Lemma \ref{lem:H}, 
it follows from Lemma \ref{lem:sets} and $S^n=I_n$ that for 
$j=0, 1,2,\dots,n-1$,
\begin{align*}
(S^j G)^n \mathcal{H}
= S^{nj} G^{n} \mathcal{H}=\mathcal{H}.
\end{align*}
Thus, we have $(S^j G)^n \in \mathcal{H}$ for $j=0, 1,2,\dots,n-1$. 
By (i) and (ii) in Lemma~\ref{lem:H}, we have 
\[
(SG)^k \mathcal{H} = S^k G^k \mathcal{H} \ne \mathcal{H} \text{ for } k=1, 2, \dots, n-1. 
\]
Therefore, $n$ is the smallest positive integer $k$ such that $(SG)^k \in \mathcal{H}$. 

Assume that $n$ is even. 
By Lemma \ref{lem:inv} we have   
\begin{align*}
&[S,G][S^2,G][S^3,G][S^4,G] \cdots[S^{n-3},G][S^{n-2},G][S^{n-1},G]\\
&=([S,G][G,S^2])([S^3,G][G,S^4]) \cdots([S^{n-3},G][G,S^{n-2}])[S^{n-1},G]\\
&=(SGS^{-1}GGS^2GS^{-2})(S^3GS^{-3}GGS^4GS^{-4})\\
&\quad \quad \cdots (S^{n-3}GS^{-n+3}GGS^{n-2}GS^{-n+2})(S^{n-1}GS^{-n+1}G)\\
&=(SGSGS^{-2})(S^3GSGS^{-4})\cdots (S^{n-3}GSGS^{-n+2})(S^{n-1}GS^{-n+1}G)\\
&=(SGSG)(SGSG) \cdots (SGSG)(SGS^{-n+1}G)\\
&=(SG)^n . 
\end{align*}
If $(SG)^n$ is the identity element of $\mathcal{H}$, we have 
\[
[S,G]=[S^2,G][S^3,G][S^4,G] \cdots[S^{n-3},G][S^{n-2},G][S^{n-1},G] 
\]
because $[S,G]$ has order 2. 
This contradicts \eqref{eq:strH-2}. 
Therefore, we see that $(SG)^n$ is not the identity element of $\mathcal{H}$. 
Since $(S^j G)^n \in \mathcal{H}$ for $j=0, 1,2,\dots,n-1$ and every element in $\mathcal{H}$ except for the identity element has order two, $2n$ is the smallest positive integer $k$ such that 
$G^k$, $(SG)^k$, $(S^2G)^k, \dots, (S^{n-1}G)^k$ are the identity element of $\mathcal{H}$. 

Assume that $n$ is odd. 
Then $G^n$ is not the identity element of $\mathcal{H}$ because $G^n = G$.   
Since $(S^j G)^n \in \mathcal{H}$ for $j=0, 1,2,\dots,n-1$ and every element in $\mathcal{H}$ except for the identity element has order two, $2n$ is the smallest positive integer $k$ such that 
$G^k$, $(SG)^k$, $(S^2G)^k, \dots, (S^{n-1}G)^k$ are the identity element of $\mathcal{H}$. 
This completes the proof. 
\hspace{\fill}$\Box$

\section*{Acknowledgement}
We acknowledge the late Professor Toyoki Matsuyama for his invaluable guidance and support during the early stages of this research. 


\end{document}